\documentclass{patmorin}
\usepackage{pat}
\usepackage{hyperref}
\usepackage{paralist}
\usepackage[noend]{algorithmic}
\usepackage{datetime2}
\hypersetup{colorlinks=true, linkcolor=linkblue,  anchorcolor=linkblue,
citecolor=linkblue, filecolor=linkblue, menucolor=linkblue,
urlcolor=linkblue, pdfcreator=Me, pdfproducer=Me} 
\title{\MakeUppercase{Near-Optimal }$O(k)$-\MakeUppercase{Robust Geometric Spanners}\thanks{%
This research was partly funded by NSERC.
}}

\author{Prosenjit Bose,\thanks{%
        School of Computer Science, Carleton University}\quad 
        Paz Carmi,\thanks{Department of Computer Science, Ben-Gurion University of the Negev}\quad
        Vida Dujmovi\'c,\thanks{%
        School of Computer Science and Electrical Engineering, University of Ottawa}\quad
        and Pat Morin\footnotemark[2]}

\date{\DTMnow}

\DeclareMathOperator{\rank}{r}
\DeclareMathOperator{\diam}{diam}
\DeclareMathOperator{\dist}{dist}
\DeclareMathOperator{\lbl}{label}
\DeclareMathOperator{\ar}{area}

\begin{document}
\maketitle

\begin{abstract}
  For any constants $d\ge 1$, $\epsilon >0$, $t>1$, and any $n$-point
  set $P\subset\R^d$, we show that there is a geometric graph $G=(P,E)$
  having $O(n\log^2 n\log\log n)$ edges with the following property:
  For any $F\subseteq P$, there exists $F^+\supseteq F$, $|F^+| \le
  (1+\epsilon)|F|$ such that, for any pair $p,q\in P\setminus F^+$,
  the graph $G-F$ contains a path from $p$ to $q$ whose (Euclidean)
  length is at most $t$ times the Euclidean distance between $p$ and $q$.
  
  In the terminology of robust spanners (Bose \etal, SICOMP,
  42(4):1720--1736, 2013) the graph $G$ is a $(1+\epsilon)k$-robust
  $t$-spanner of $P$. This construction is sparser than the recent
  constructions of Buchin, Ol\`ah, and Har-Peled (arXiv:1811.06898)
  who prove the existence of $(1+\epsilon)k$-robust $t$-spanners with
  $n\log^{O(d)} n$ edges.
\end{abstract}

\section{Introduction}

A geometric graph $G=(P,E)$ with vertex set $P\subset\R^d$ is a (geometric)
$t$-spanner of a subset $X\subseteq P$ if, for every pair of distinct vertices
$p,q\in X$, 
\begin{equation}
  \frac{\dist_G(p,q)}{\dist(p,q)} \le t \enspace , \eqlabel{spanning-ratio}
\end{equation}
where $\dist(p,q)$ denotes the Euclidean distance between $p$ and
$q$ and $\dist_G(p,q)$ denotes the Euclidean length of the shortest
path between $p$ and $q$ in $G$, where we use the convention that
$\dist_G(p,q)=\infty$ if $p$ and $q$ are in different components of $G$.
Most of the research on spanners focuses on \emph{sparse} spanners,
where the number of edges in $G$ is linear, or close to linear, in $|P|$.
In addition to having natural applications to transportation networks,
sparse $t$-spanners have found numerous applications in communication
networks, approximation algorithms, and geometric data structures.
A book \cite{ns07} and handbook chapter \cite{e99} provide extensive
discussions of geometric $t$-spanners and their applications.

For any non-decreasing function $f\colon\N\to\N$, Bose \etal\
\cite{bose.dujmovic.ea:robust} define a geometric graph $G=(P,E)$ to be an
\emph{$f(k)$-robust $t$-spanner} if, for every set $F\subseteq P$,
there exists a set $F^+\supseteq F$ such that $|F^+|\le f(|F|)$ and
the graph $G-F$ is a $t$-spanner of $V(G)\setminus F^+$.  In networking
applications, this definition captures the idea that the number
of nodes ($|F^+|$) harmed by a set of faulty nodes ($F$) should be bounded by a function
($f$) of the number of faulty nodes ($|F|$), independent of the network
size ($|P|$).

Under this definition, the most robust spanner one could hope for
is a $k$-robust spanner, but it is straightforward to argue
that, even for one-dimensional point sets, the complete graph is the
only $k$-robust spanner.\footnote{Proof: Consider any pair of vertices
$v,w\in V(G)$ that are not adjacent in $G$ and let $F=V\setminus\{v,w\}$.
Then $\|vw\|_{G-F}=\infty$ so $G-F$ is a not a $t$-spanner of $V\setminus
F^+=V\setminus F$ for any $t<\infty$.} The complete graph is not sparse,
and is therefore not suitable for many applications.

A natural second-best option is a $(1+\epsilon)k$-robust spanner with
a near-linear number of edges, for some small constant $\epsilon >0$.
Buchin \etal\ \cite{buchin.har-peled.ea:spanner} call these objects
\emph{$\epsilon$-resilient} spanners and prove the existence of
$\epsilon$-resilient spanners with $n\log^{O(d)} n$ edges.
In the current paper we reduce the dependence on $d$ by proving the
following theorem:

\begin{thm}\thmlabel{main-i}
  For every constant $d\ge 1$, $\epsilon>0$, $t>1$ and every $n$-point
  set $P\subseteq\R^d$, there exists an $\epsilon$-resilient $t$-spanner
  $G=(P,E)$ with $|E|=O(n\log^2 n\log\log n)$.
\end{thm}

Bose \etal\ \cite{bose.dujmovic.ea:robust} show that, for any constants
$\epsilon>0$ and $t\ge 1$, there exists $1$-dimensional point sets
for which any $(1+\epsilon)k$-robust $t$-spanner has $\Omega(n\log
n)$ edges.  Thus, \thmref{main-i} is within a factor of $O(\log
n\log\log n)$ of optimal in any constant dimension.  (Note that in
dimension $d=1$ optimal constructions, having $O(n\log n)$ edges, are
known \cite{buchin.har-peled.ea:spanner}.)

The proof of \thmref{main-i} uses several ingredients: The well-separated
pair decomposition \cite{callahan.kosaraju:decomposition}, which
is fairly standard in spanner constructions.  Expander graphs
\cite{hoory.linial.ea:expanders}, that are a natural tool to
achieve robustness. Two less obvious techniques we use are a centroid
decomposition (i.e., hierarchical balanced separators) for binary trees
and an old idea of Willard \cite{willard:maintaining} for file maintenance
(aka, order maintenance) that involves a hierarchical structure whose
smaller substructures have more stringent density requirements than
larger substructures.

These last two ideas represent a significant departure from the work of
Buchin \etal\ \cite{buchin.har-peled.ea:spanner} who 
(among other tools) also use well-separated
pair decompositions and expanders.  Their constructions, of
which there are two, rely on a reduction to the 1-dimensional problem
and the fact that the paths obtained in the 1-d case have $O(\log n)$ edges.
However, they have very little fine-grained control over the lengths of
these edges, which requires them to construct a $d$-dimensional object
($\theta$-graphs  \cite{keil.gutwin:classes} or locality-preserving orderings \cite{chan.har-peled.ea:on}) in which the relevant
parameter ($\theta$ and $\varsigma$, respectively) is $O(1/\log n)$.  This leads
to $\log^{O(d)} n$ factors in the number of edges in their constructions.

In the remainder of the paper we first review some relevant background material and then present our $\epsilon$-resilient spanner construction.

\section{Background}

In this section we briefly review some existing results used in our
construction.

\subsection{Expanders}

For a graph $G$ and a vertex $x\in V(G)$, define the \emph{neighbourhood}
of $x$ in $G$ as $N_G(x) = \{ y: xy\in E(G)\}$.  For a subset $X\subseteq
V(G)$, $N_G(X)=\bigcup_{x\in X} N_G(x)$.  For a subset $Y\subseteq
V(G)$, define the \emph{shadow} of $Y$ in $G$ as $S_G(Y) = \{x\in V(G):
N_G(x)\subseteq Y\}$.

Results like the following lemma, and its proof, are fairly standard
expander constructions (see, for example, the survey by Hoory \etal\
\cite{hoory.linial.ea:expanders}):

\begin{lem}\lemlabel{expand}
   For any $k\ge 2$, $\ell\ge 2$, $n\in\N$ and any two sets $A$ and $B$
   each of size $\Theta(n)$, there exists a graph $H=(A\cup B,E)$
   with $|E|=O(n(k\log \ell + \ell\log k))$ such that, for any set $B'\subset B$, $|B'|\ge |B|/\ell$, \[ |N_H(B')| \ge (1-1/k)|A| \enspace . \]
\end{lem}

\begin{proof}
  For simplicity of calculation, assume that $|A|=|B|=n$.  Fix some
  subset $A'\subset A$ of size $|A'|=(1-1/k)|A|$.  Let $a_1,\ldots,a_r$
  be a sequence of $r$ iud random samples from $A$.  Then the probability
  that all of these samples are in $A'$ is
  \[
     \Pr\{\{a_1,\ldots,a_r\}\subset A'\} = (|A'|/|A|)^r = (1-1/k)^r \le e^{-r/k}
  \]
  Let $A$ and $B$ be disjoint $n$-element sets and construct a random
  graph $H$ where each element in $B$ forms an edge with $\Delta$ randomly
  chosen (with replacement) elements in $A$.  For a fixed $A'\subset A$
  with $|A'|=(1-1/k)n$ and a fixed $B'\subset B$ with $|B'| = n/\ell$,
  \[
    \Pr\{N_H(B') \subseteq A'\} 
        \le (1-1/k)^{\Delta n/\ell} 
        \le e^{-\frac{\Delta n}{k\ell}}
  \]
  Let $\mathcal{E}$ be the event that there exists $A'\subset
  A$, $|A'|=(1-1/k)n$, $B'\subset B$, $|B'|=n/\ell$ such that
  $N_H(B')\subseteq A'$.  Then
  \begin{align*}
    \Pr\{\mathcal{E}\} 
        & \le \binom{n}{(1-1/k)n}\binom{n}{n/\ell}e^{-\frac{\Delta n}{k\ell}} \\
        & = \binom{n}{n/k}\binom{n}{n/\ell}e^{-\frac{\Delta n}{k\ell}} \\
        & \le (ek)^{n/k} (e\ell)^{n/\ell}e^{-\frac{\Delta n}{k\ell}} \\
        & = \exp((n/k)(1+\ln k) + (n/\ell)(1+\ln(\ell)) - (\Delta n)/(k\ell)) \\
        & < 1
  \end{align*}
  for $\Delta > k(1+\ln \ell) + \ell(1+\ln k)$.  In particular, there
  must exist at least one graph with $O(n(k\log\ell + \ell\log k))$
  edges that satisifies the conditions of the lemma.
\end{proof}

\lemref{expand} can be interpreted informally as saying that even small
subsets of $B$ (of size at least $n/\ell$) have neighbourhoods that
expand into most of $A$.  The following lemma, expressed in terms of
shrinking shadows of subsets of $A$, is also useful:

\begin{lem}\lemlabel{shrink}
   For any $k\ge 2$, $\tau\ge 1$ and any two sets $A$ and $B$ with
   $|A|=\Omega(|B|)$, there exists a graph $H=(A\cup B,E)$ with $|E|=O(|B|(k\log
   \tau + \tau\log k + \log|B|))$ such that for any $A'\subset A$ with $|A'|\le
   (1-1/k)|A|$,
   \[ |S_H(A')| \le |A'|/\tau \enspace .\]
\end{lem}

The proof of \lemref{shrink} is similar to, but somewhat more involved than, the proof of \lemref{expand}.
Each element in $B$ chooses $\Delta$ random neighbours in $A$.  Then, one shows that, for each $x\in\{1,\ldots,\min\{|B|, (1-1/k)|A|/\tau\}\}$,
\[
    \binom{n}{\tau x}\binom{|B|}{x}(\tau x/n)^{\Delta x} < 1/|B|
\]
for some $\Delta=O(k\log\tau +\tau\log k + \log|B|)$.

\subsection{Fair-Split Trees and Well-Separated Pair Decompositions}

For two points $p,q\in\R^d$, $\dist(p,q)$, denotes the Euclidean distance
between $p$ and $q$. For two sets $P,Q\subset\R^d$, the distance between
$P$ and $Q$ is $\dist(P,Q)=\min\{\dist(p,q):p\in P, q\in Q\}$.  For a
point set $P\subset\R^d$, the \emph{diameter} of $P$ is denoted by
$\diam(P)=\max\{\dist(p,q):p,q\in P\}$.

For a rooted binary tree $T$, $L(T)$ denotes the set of leaves in
$T$. We use the convention that, if $T$ consists of a single node $u$,
then $L(T)=\{u\}$. The \emph{size} of $T$, denoted $|T|$ is the number of
leaves, $|L(T)|$, of $T$. For a node $u$ in $T$, $T_u$ denotes the subtree
of $T$ rooted at $u$.  We say that $T$ is \emph{full} if each non-leaf
node of $T$ has exactly two children.

A \emph{fair-split tree} $T$ is a full binary tree whose leaves are
points in $\R^d$ and whose nodes have the following \emph{fair-split
property}: We let $R(T)$ denote the minimum axis-aligned bounding
box of $L(T)$ and we let $\diam'(T)$ denote the sum of the side
lengths of $R(T)$.  For any node $w$ with parent $x$, $\diam'(T_w)
\le (1-1/(2d))\diam'(T_x)$.\footnote{Traditionally, fair-split trees
are described as splitting $R(x)$ by bisecting its longest side.
This obviously implies that $\diam'(w)\le 1-(1/(2d))\diam'(x)$.} It is
worth noting that $\diam(L(T))$ and $\diam'(T)$ are bounded by each other:
\[
	\diam(L(T)) \le \diam'(T) \le d\cdot\diam(L(T)) \enspace .
\]	
For any $n$-point set $P\subset\R^d$, a fair-split tree for $P$ can be
computed in $O(dn\log n)$ time \cite{callahan.kosaraju:decomposition}.

For a finite point set $P\subset\R^d$ and any $s>0$, a
\emph{well-separated pair decomposition (WSPD)} of $P$ is a set of pairs
$\{(A_i,B_i):i\in\{1,\ldots,m\}\}$ with the following properties:
\begin{enumerate}
  \item For every $i\in\{1,\ldots,m\}$, 
    $\dist(A_i,B_i)\ge s\cdot\max\{\diam(A_i),\diam(B_i)\}$.
  \item For every pair $p,q\in P$ there exists exactly one
    $i\in\{1,\ldots,m\}$ such that $p\in A_i$ and $q\in B_i$, or $q\in A_i$
    and $p\in B_i$.
\end{enumerate}
Well-separated pair decompositions were introduced by Callahan and
Kosaraju \cite{callahan.kosaraju:decomposition}, who construct them
using fair-split trees.

\begin{thm}[Callahan and Kosaraju 1995]\thmlabel{wspd}
  For any constant $d\ge 1$, any $s\ge 1$ and any $n$-point set
  $P\subset\R^d$ with fair split tree $T=T(P)$, there exists a WSPD
  $\{(A_i,B_i):i\in\{1,\ldots,m\}\}$ of $P$ with size $m\in O(s^d n)$.
  Furthermore, each pair $(A_i,B_i)=(L(T_{a_i}),L(T_{b_i}))$ where $a_i$
  and $b_i$ are nodes of $T$.
\end{thm}
We call the WSPD guaranteed by \thmref{wspd} a WSPD of $P$ \emph{using}
$T$.  In his thesis, Callahan proves an additional useful result about
well-separated pair decompositions \cite[Section~4.5]{callahan:dealing}:

\begin{lem}[Callahan 1995]\lemlabel{wspd-ii}
  In the WSPD of \thmref{wspd},
   $\sum_{i=1}^m\min\{|A_i|,|B_i|\} = O(s^d n\log n)$.
\end{lem}

\section{The Construction}

In this section we describe our $\epsilon$-resilient
$t$-spanner construction for an $n$-point set $P\subset\R^d$.
Throughout this section, $T$ is a fair-split tree for $P$ and
$W=\{(A_i,B_i):i\in\{1,\ldots,m\}\}$ is an $s$-well-separated pair
decomposition for $P$ using $T$.

Our goal is to construct a $(1+\epsilon)$-robust $t$-spanner $G$
of $P$.  At times we will deal with an arbitrary subset $F\subseteq
P$ whose elements are \emph{faulty} and we will identify a superset
$F^+\supseteq F$ so that $G-F$ is a $t$-spanner of $P\setminus F^+$.
We will say that the elements of $F^+\setminus F$ are \emph{abandoned}.
Our construction of $G$ proceed in two phases, which are described
in Sections~\ref{sec:exploding} and \ref{sec:navigating}, respectively:

\begin{enumerate}
  \item We first show how to construct a graph $G_T$ where most of the
  points in $p\in P\setminus F$ have the following property: For every
  ancestor $u$ of $p$ in $T$, and for the vast majority of points $q\in
  L(T_u)\setminus F$, $G_T-F$ contains a path from $p$ to $q$ of length
  $O(\diam'(T_u))$.  When a point $p\in P\setminus F$ has this property,
  we say that $p$ is able to \emph{explode into} $u$.

  We define $F^+_T$ as the set of points $p\in P$ such that $p$ has an
  ancestor $u$ and $p$ is \emph{un}able to explode into $u$.  We will
  have to abandon the points in $F^+_T$.  The parameters of $G_T$ are
  chosen so that $F^+_T\setminus F$ has size at most $\epsilon|F|$.
  and the graph $G_T$ has $O(n\log^2 n\log\log n)$ edges.

   \item Next, we consider the WSPD $W$. For points $p\in A_i\setminus
   F^+_T$ and $q\in B_i\setminus F^+_T$ to have a spanner path between
   them in $G-F$, it is sufficient that $G$ contains an edge $qq'$
   with $q'\in A_i\setminus F^+_T$.  This is because both $p$ and $q'$
   can explode into $A_i$ and the subsets of $A_i\setminus F$ that $p$
   and $q'$ can explode into are so large that they must have at least
   one edge joining them.

   The third step in our construction, therefore, is to consider each pair
   $(A_i,B_i)\in W$ with $|A_i|\ge |B_i|$ and construct an expander $H$
   from $B_i$ into $A_i$.  By doing this carefully, we ensure that
  \begin{enumerate}
     \item $F^+_T$ contains most of $A_i$; or
     \item $S_{H}(F^+_P)$ is much smaller than $F^+_P\cap A_i$.
  \end{enumerate}
  In the former case, we can abandon all the points in $A_i$ without
  abandoning too many additional points. In the latter case, we can
  abandon the points in $S_{H}(F^+_P)$ without abandoning too many additional 
  points.
\end{enumerate}

\subsection{Exploding in the Fair-Split Tree}
\seclabel{exploding}

In this section we describe the graph $G_T$ that allows most of the
points in $P\setminus F$ to explode into any of their ancestors in $T$.

Consider the following recursively constructed graph $G_{T}$ whose vertex set
is $P=L(T)$.  If $|T| \le \kappa$ for some constant $\kappa$, then $G_T$
is the complete graph on $L(T)$. For our particular application, we will
choose $\kappa\ge 5$.  
If $|T|>\kappa$, let $u_0$ be a node of $T$ with the property that
$|T|/3\le |T_{u_0}|\le 2|T|/3$.  The existence of $u_0$ (or rather
the edge from $u_0$ to its parent) is a standard result on binary
trees.\footnote{Proof: Begin by setting $v_0$ to the root of $T$ and then
repeatedly set $v_{i+1}$ to be the child of $v_i$ whose subtree contains
at least half the leaves of $T_{v_i}$.  The smallest index $i$ for which
$|T_{v_i}|\le 2|T|/3$ yields the desired node $u_0=v_i$.} 

Fundamental to the analysis in this section is the \emph{rank} of a
tree $T$, defined as $\rank(T)=\floor{\log_{3/2}|T|}$.  Note that, for
$|T|\ge\kappa\ge 5$, $\rank(T_{u_0}) \ge \floor{\log_{3/2}(5/3)} \ge 1$.

Let $T_1$ be the full binary tree obtained from $T-T_{u_0}$ by contracting
an edge incident to the unique non-leaf node of $T-T_{u_0}$ that has only
one child.  The graph $G_{T}$ contains an expander $H_T=(L(T),E_T)$. This
expander has parameters $\Delta>1$, $\alpha, \beta,\zeta,\eta > 0$ and is
constructed so that it satisfies the following properties:
\begin{enumerate}
  \item[(PR1)] For any $X\subset L(T_{u_0})$ with
    $|X|<(1-\beta/\Delta)|T_{u_0}|$, 
    \[ |S_{H_T}(X)|\le (\alpha/\Delta)|X| \enspace . \]

  \item[(PR2)] For any set $Y\subset L(T)$ with $|Y|\ge
    (\zeta/\Delta)L(T)$, \[ |N_{H_T}(Y)|\ge (1-\eta/\Delta)|T| \enspace .\]
\end{enumerate}
Informally, Property~(PR1) tells us that, if some subset $X$ of
$T_{u_0}$ becomes disabled, then this only prevents a much smaller subset
$S_{H_T}(X)$ of $T_{1}$ from accessing $T_{u_0}$.  Property~(PR2) tells
us that if some point $p$ can reach a $\zeta/\Delta$ fraction of the
points in $L(T)$ then $p$ can reach nearly all the points in $L(T)$.

In our construction, $\Delta=\ceil{\log_{3/2} n}$ and the remaining
parameters are small values that are upper bounded by some function
of $\epsilon$. In particular, for any constant $\epsilon >0$, these
parameters are also constant. Note that we distinguish here between $n$
and $|T|$. This is because, in recursive calls $\Delta=\ceil{\log_{3/2} n}$
remains fixed even though the recursive input has size smaller than $n$.

After constructing $H_T$, we recursively construct $G_{T_{u_0}}$
and $G_{T_1}$ and add the edges of each of the resulting graphs to
$G_{T}$. This concludes the description of the graph $G_T$.

\begin{clm}
  For any constants $\alpha,\beta,\zeta,\eta>0$, there exists a graph $H_T$ with
  $O(|T|\Delta\log\Delta)$ edges that satisfies Properties~(PR1) and (PR2).
\end{clm}

\begin{proof}
  To satisfy Property~(PR1), $H_T$ contains an expander described
  by \lemref{shrink} for the pair $(A=L(T_{u_0}),B=L(T_1))$ with
  parameter $k=\Delta/\beta$ and $\tau=\Delta/\alpha$.  This graph has
  $O(|T|\Delta\log\Delta)$ edges.

  To satisfy Property~(PR2), $H_T$ contains an expander described by
  \lemref{expand} for the pair $(A=L(T),B=L(T))$ with
  parameters $k=\Delta/\eta$ and $\ell=\Delta/\zeta$. This graph also
  has $O(|T|\Delta\log\Delta)$ edges.
\end{proof}

\begin{clm}\clmlabel{edges-t}
  The graph $G_{T}$ has $O((\Delta\log\Delta)|T|\log |T|)$ edges.
\end{clm}

\begin{proof}
  The graph $H_T$ has
  $O(|T|\Delta\log\Delta)$ edges.  The recursive constructions are on two trees $T_{u_0}$
  and $T_1$ where $|T_{u_0}|+|T_1|=|T|$ and $\max\{|T_{u_0}|,|T_1|\}\le
  2|T|/3$. It follows that the depth of recursion is at most
  $\log_{3/2}|T|$ and each level of recursion contributes a total of
  $O((\Delta\log\Delta)|T|)$ edges for a total of $O((\Delta\log\Delta)|T|\log|T|)$ edges.
\end{proof}

Recall that $\rank(T)=\floor{\log_{3/2} |T|}$ and observe that, in
the preceding construction, $\rank(T_{u_0}) \le \rank(T)-1$ and
$\rank(T_1)\le\rank(T)-1$.  Let $F$ be an arbitrary subset of $P$.  We say
that $T$ is \emph{$F$-dense} if $|L(T)\cap F|\ge (1-\delta\rank(T)/\Delta)|T|$
for some constant $\delta$ to be discussed shortly.  Define the set
$F^+_T$, recursively, as follows (here $u_0$ and $T_1$ are defined as above):

\begin{enumerate}
  \item If $T$ is $F$-dense, then $F^+_T\gets L(T)$. 
     \newline [Too many points in $T$ are faulty so we abandon all points in $T$.]
  \item $F^+_T\gets F^+_{T_{u_0}}\cup F^+_{T_1}$. 
     \newline [Recursive computation on $T_{u_0}$ and $T_1$.]
  \item If $|F^+_{T_{u_0}}|\le (1-\beta/\Delta)|T_{u_0}|$
  \begin{enumerate}
     \item then $F^+_T\gets F^+_T\cup S_{H_T}(F^+_{T_{u_0}})$. 
         \newline [Abandon points in $T_1$ that can only reach $T_{u_0}$ through points abandoned in Step~2.]
     \item Otherwise, $|F^+_{T_{u_0}}|> (1-\beta/\Delta)|T_{u_0}|$,
          and $F^+_T\gets F^+_T\cup L(T_{u_0})$.
        \newline [Too many points in $T_{u_0}$ are already abandoned, abandon everything in $T_{u_0}$.]
  \end{enumerate}
\end{enumerate}

Below, we claim that $|F^+_T|\le (1+\epsilon \rank(T)/\Delta)|F|$, for
some small $\epsilon >0$.  Before diving into the proof, we first give an
informal sketch.  In Step~1, the definition of $F$-density ensures that,
if $T$ is $F$-dense, then it safe to discard all of $T$.  By induction,
Step~2 obviously produces a sufficiently small set $F^+_T$.  In fact,
since $\rank(T_{u_0})$ and $\rank(T_1)$ are both smaller than $\rank(T)$,
Step~2 produces a set that is smaller than necessary. Specifically, at this
point we can afford to add an additional $(\epsilon/\Delta)|F\cap L(T)|$
elements to $F^+_T$.  The condition in Step~3 ensures that, in either
of the two cases, the number of elements we add to $F^+_T$ is, indeed,
at most $(\epsilon/\Delta)|F\cap L(T)|$.  In Step~3(a), we know that
the shadow of $F^+_{T_{u_1}}$ in $T_1$ is sufficiently small to add it
to $F^+_T$.  In Step~3(b) we know that $F^+_{T_{u_0}}$ is so large that
we can add the rest of it to $F^+_T$.

\begin{clm}\clmlabel{f-plus}
   For any constant $\epsilon>0$ there are constants
   $\alpha,\beta,\zeta,\eta >0$ such that, for any $F\subseteq P$,
   $|F^+_T| \le (1+\epsilon\rank(T)/\Delta)|F\cap L(T)|$.
\end{clm}

\begin{proof}
  The proof is by induction on $\rank(T)$. If $|T|=1$, the claim is
  obvious. For $|T|\ge 2$, there are two cases to consider:
  \begin{enumerate}
    \item $T$ is $F$-dense. In this case $F^+_T=L(T)$.  
     Since $T$ if $F$-dense,  $|L(T)\cap F|\ge
     (1-\delta\rank(T)/\Delta)|T|$.  So
     \[
       |F^+_T|=|T|
	  \le \frac{|L(T)\cap F|}{1-\delta\rank(T)/\Delta} 
          \le (1+\epsilon\rank(T)/\Delta)|L(T)\cap F|
     \]
     provided that $\epsilon \ge 1/(1-\delta)-1$ (e.g., $\delta\le \epsilon/2$).

    \item $T$ is not $F$-dense. There are two subcases to consider:
    \begin{enumerate}
       \item $|F^+_{T_{u_0}}| \le (1-\beta/\Delta)|T_{u_0}|$.
         In this case, $F^+_T=F^+_{T_{u_0}}\cup F^+_{T_1} \cup S_{H_T}(F^+_{T_{u_0}})$.
         Recall that
        $\rank(T_{u_0}),\rank(T_1)\le\rank(T)-1$ so, by induction,
         \begin{align}
              |F^+_{T_{u_0}}| + |F^+_{T_{1}}| 
              & \le (1+\epsilon\rank(T_{u_0}))/\Delta)|F\cap L(T_{u_0})| 
                    +(1+\epsilon\rank(T_{1}))/\Delta)|F\cap L(T_{1})| \notag \\
              & \le (1+\epsilon(\rank(T)-1))/\Delta)|F\cap L(T_{u_0})| 
                    +(1+\epsilon(\rank(T)-1))/\Delta)|F\cap L(T_{1})|  \notag \\
              & = (1+\epsilon\rank(T)/\Delta)|F\cap L(T)| - (\epsilon/\Delta)|F\cap L(T)| \enspace .
            \eqlabel{inductio}
         \end{align}
         All that remains is to show that $|S_{H_T}(F^+_{T_{u_0}})|\le
         (\epsilon/\Delta)|F\cap L(T)|$. By Property~(PR1) of $H_T$,
          \begin{align*}
          |S_{H_T}(F^+_{T_{u_0}})| 
            & \le (\alpha/\Delta)|F^+_{T_{u_0}}| \\
            & \le (\alpha/\Delta)(1+\epsilon\rank(T_{u_0})/\Delta)|F\cap L(T_{u_0})| \\
            & \le (\alpha/\Delta)(1+\epsilon\rank(T_{u_0})/\Delta)|F\cap L(T)|\\
            & = (\alpha/\Delta+\alpha\epsilon\rank(T_{u_0})/\Delta^2)|F\cap L(T)| \\
            & \le (\alpha/\Delta+\alpha\epsilon/\Delta)|F\cap L(T)| 
    	     & \text{(for $\rank(T_{u_0})/\Delta\le 1$)} \\
    	& \le (\epsilon/\Delta)|F\cap L(T)| \enspace ,
       \end{align*}
       provided that $\alpha+\alpha\epsilon \le \epsilon$, i.e, 
       $\alpha \le \epsilon/(\epsilon+1)$.   

       \item $|F^+_{T_{u_0}}| > (1-\beta/\Delta)|T_{u_0}|$. In this case, $F^+_T=L(T_{u_0}) \cup F^+_{T_1}$ and
       \begin{align*}
          |F^+_T| 
            & = |T_{u_0}| + |F^+_{T_1}| \\
            & \le (1+2\beta/\Delta)|F^+_{T_{u_0}}| + |F^+_{T_1}|
              & \text{(for $\beta/\Delta\le 1/2$)}\\
            & = |F^+_{T_{u_0}}| + |F^+_{T_1}| + (2\beta/\Delta)|F^+_{T_{u_0}}| \\
            & \le (1+\epsilon\rank(T)/\Delta)|F\cap L(T)| - (\epsilon/\Delta)|F\cap L(T)| + (2\beta/\Delta)|F^+_{T_{u_0}}| 
             & \text{(as in \eqref{inductio})} \\
            & \le (1+\epsilon\rank(T)/\Delta)|F\cap L(T)| - (\epsilon/\Delta)|F\cap L(T)| + (4\beta/\Delta)|F\cap L(T_{u_0})|  \\
              & \qquad \text{(since $|F\cap L(T_{u_0})| \ge |F^+_{T_{u_0}}|/(1+\epsilon\rank(T_{u_0})/\Delta)\ge|F^+_{T_{u_0}}|/2$)}\\
            & \le (1+\epsilon\rank(T)/\Delta)|F\cap L(T)| - (\epsilon/\Delta)|F\cap L(T)| + (4\beta/\Delta)|F\cap L(T)| \\
              & \qquad \text{(since $L(T_{u_0})\subseteq L(T)$)} \\
            & \le (1+\epsilon\rank(T)/\Delta)|F\cap L(T)| \enspace ,
       \end{align*}
       provided that $\beta \le \epsilon/4$. \qedhere
     \end{enumerate}
   \end{enumerate}
\end{proof}

Next we show that any point of $P$ not in $F^+_T$ can explode into any of its
ancestors in $T$.

\begin{clm}\clmlabel{explode}
  Let $C=4d$ and let $a=\beta/2$.  For any node $u$ of $T$ and
  every point $p\in L(T_u)\setminus F^+_T$, there exists $X\subset L(T_u)$,
  $|X|\ge (1-a/\Delta)|T_u|-|F^+_T\cap L(T_u)|$ such that for every $q\in X$,
  $G_T-F$ contains a path from $p$ to $q$ of length at most $C\diam(T_u)$,
  for $C\le 4d$.
\end{clm}

\begin{proof}
  The proof is by induction on $|T|$.  If $|T|\le\kappa$, the result
  is trivial since $G_T$ is the complete graph.  Thus, we may assume
  that $|T|>\kappa$.  If $T_u$ is contained in $T_{u_0}$ then we can
  apply induction on $T_{u_0}$ (with $T=T_{u_0}$ and $u=u$).   This yields a set $X$ of size
  \[
     (1-a/\Delta)|T_u|-|F^+_{T_{u_0}}\cap L(T_u)|
     \ge (1-a/\Delta)|T_u|-|F^+_{T}\cap L(T_u)| \enspace .
  \]
  as required. Similarly, if $T_u$ is contained in $T_1$ then we can
  apply induction on $T_1$.

  The only possibility that remains is that $u$ is a proper ancestor
  of $u_0$, so $T_u$ contains leaves of $T_{u_0}$ and leaves of $T_1$.
  There are several cases to consider:
  \begin{enumerate}
    \item $|F^+_{T_{u_0}}|\le (1-\beta/\Delta)|T_{u_0}|$. In this case, there are two subcases
    to consider:
    \begin{enumerate}
      \item $p\in L(T_{u_0})$.  Since $u_0$ is not the root of $T$,
      $\diam'(T_{u_0}) \le (1-1/2d)\diam'(T_u)$.  We apply
      induction on $T_{u_0}$ (with $T=T_{u_0}$ and $u=u_0$)
       to find a $p$-reachable set $X_0\subseteq L(T_{u_0})$ of size
      \begin{align*}
        |X_0| & \ge (1-a/\Delta)|T_{u_0}|-|F^+_{T_{u_0}}| \\
              & \ge (1-a/\Delta)|T_{u_0}|-(1-\beta/\Delta)|T_{u_0}| \\
              & = ((\beta-a)/\Delta)|T_{u_0}| \\
              & = (\beta/(2\Delta))|T_{u_0}|
              & \text{(since $a = \beta/2$).} \\
              & \ge (\beta/(6\Delta))|T| \enspace .
              & \text{(since $|T|/3\le |T_{u_0}|$).} 
      \end{align*}
      By Property~(PR2) of $H_T$ (with $\zeta = \beta/6$ and $\eta =
      a/3$), we can then take $X=N_{H_T}(X_0)\cap L(T_u)\setminus F^+_T $.  Then
      \begin{align*}
         |X| & \ge (1-\eta/\Delta)|T|-|F^+_T\cap L(T_u)| - (|T|-|T_u|) \\
             & = |T_u|-(\eta/\Delta)|T|-|F^+_T\cap L(T_u)| \\
             & \ge (1-3\eta/\Delta)|T_u|-|F^+_T\cap L(T_u)| 
                & \text{(since $3|T_u|\ge |T|$)}\\
             & = (1-a/\Delta)|T_u|-|F^+_T\cap L(T_u)| \\
      \end{align*} 
      and every point $q\in X$ is reachable from $p$ by a path in $G_T-F$
      of length at most
        \[ (C(1-(1/2d))+1)\diam'(T_{u}) < C\diam'(T_u) \]
      for $C= 4d$.

      \item $p\in L(T_1)$. Since $p\not\in F^+_T$, $H_T$ contains an edge
      from $p$ to some point $p'\in L(T_{u_0})\setminus F^+_{T_{u_0}}$.
      As described in the previous case, there is a set $X\subset
      L(T)$ of size $(1-a/\Delta)|T|-|F^+_{T}|$ that is reachable
      from $p'$ by paths of length at most $((1-1/2d)C+1)\diam'(T_u)$.
      The edge $pp'$ has length at most $\diam'(T_u)$. Therefore every
      $q\in X$ is reachable from $p$ using paths of length at most
      $((1-1/2d)C+2)\diam'(T_u) = C\diam'(P)$ for $C= 4d$.
    \end{enumerate}

    \item $|F^+_{T_{u_0}}|> (1-\beta/\Delta)|T_{u_0}|$.  In this case,
     $F^+_T= L(T_{u_0})\cup F^+_{T_1}$.
     Therefore, $p\in L(T_1)$.  Now, we apply
     induction on $T_1$ (with $T=T_1$ and $u=u$) and obtain a set $X\subseteq L(T_u)$
     that can be reached by $p$ in $G_T-F$ with paths of length at most
     $C\diam'(T_u)$.  Now,
     \begin{align*}
       |X| & \ge (1-a/\Delta)|(T_1)_u| - |F^+_{T_1}\cap L((T_1)_u)| \\
           & =   (1-a/\Delta)|(T_1)_u| - |F^+_{T_1}\cap L((T_1)_u)|
                    - |T_{u_0}| + |T_{u_0}| \\
           & \ge   (1-a/\Delta)|(T_1)_u| - |F^+_{T}\cap L((T_1)_u)|
                    - |T_{u_0}| + |T_{u_0}| \\
           & =   (1-a/\Delta)|(T_1)_u| - |F^+_{T}\cap L(T_u)| + |T_{u_0}| \\
             & \qquad \text{(since $F^+_T\cap L(T_u)= L(T_{u_0}) \cup (F^+_T\cap L((T_1)_u))$)} \\
           & >   (1-a/\Delta)|T_u| - |F^+_{T}\cap L(T_u)| 
             & \text{(since $|T_u|=|(T_1)_u|+|T_{u_0}|$),} 
     \end{align*}
     as required. \qedhere
    \end{enumerate}
\end{proof}

\subsection{Navigating the Well-Separated Pairs}
\seclabel{navigating}

Assume $\epsilon < 1/2$.  For each node $u$ of $T$, define
$\lbl(u)=\floor{\log_{1+\epsilon}|T_u|}$. We say that a node $u$ of $T$
is \emph{special} if $u$ is a leaf or if $\lbl(u)$ is different from both
its children.  If $u$ is special, then $T_u$ is also \emph{special}.
Observe that for every node $w$ of $T$, $T_w$ contains a special
subtree $T_{w'}$ with $|T_{w'}|\ge (1-2\epsilon)|T_w|$.\footnote{Proof:
Consider the subtree $T_w'$ of $T_w$ induced by all nodes $v$ in $T_w$
such that $\lbl(v)=\lbl(w)$. $T_w'$ is non-empty and therefore contains
at least one leaf $w'$ which, by definition, is a special node of $T$.
Therefore $T_{w'}$ is a special subtree of $T_w$ and $|T_{w'}|\ge
|T_w|/(1+\epsilon)$ which implies $|T_{w'}|\ge (1-2\epsilon)|T_w|$, for
$\epsilon \le 1/2$.}  Let $S(T)$ denote the set of special nodes in $T$.

Recall that $W=\{(A_i,B_i):i\in\{1,\ldots,m\}\}$ is an $s$-well-separated
pair decomposition for $P$ using $T$. For each $i\in\{1,\ldots,m\}$,
we use the following notational conventions:  $|A_i|\ge |B_i|$,
$A_i=L(T_{a_i})$, and $B_i=L(T_{b_i})$, where $a_i$ and $b_i$ are nodes
of $T$.

Our robust spanner begins with the graph $G_T$ described in the
previous section that is constructed using the fair-split tree $T$.
Next, we will add a new graph $G_W$ that provides connections between
the well-separated pairs in $W$.

To begin, we create a new set of well-separated pairs $W'$ as follows:
For each pair $(A_i,B_i)\in W$, we find the largest special subtree
$T_{a_i'}$ of $T_{a_i}$ and add the pair $(A_i',B_i)=(L(T_{a_i'}),B_i))$
to $W'$.  [Although each pair $(A_i',B_i)\in W'$ is $s$-well-separated,
$W'$ is not necessarily a WSPD for $P$.  In particular, there are pairs
of points with $p\in A_i\setminus A_i'$ and $q\in B_i$ that
are not represented in $W'$.]

Next, we partition $W'$ into groups $\{W'_u: u\in S(T)\}$ indexed by
the special nodes of $T$ where, for each $u\in S(T)$,
\[
	W'_u = \{ (A_i',B_i)\in W' : a_i' = u \}
\]
For each group $W'_u$, define $B'_u=\bigcup_{(A,B)\in W'_u} B$
and let $H'_u$ be an expander graph on the pair $(L(T_u), B'_u)$ with
the following properties:
\begin{enumerate}
  \item[(PR3)] For any $X\subseteq L(T_u)$ with $|X|\le
  (1-\epsilon)|T_u|$, 
   \[ |S_{H'_u}(X)| \le \epsilon|X|/\log_{1+\epsilon} n \]
%
  \item[(PR4)] For any two sets $X,Y\subset L(T_u)$ with $|X|,|Y|\ge
  \epsilon|T_u|$, $G$ contains at least one edge $xy$ with $x\in X$
  and $y\in Y$.
\end{enumerate}

\begin{clm}\clmlabel{edge-count-iii}
  There exists a graph $H'_u$ that satisifies Properties~(PR3) and (PR4)
  that has $O((|T_u|+|B'_u|)\log n\log\log n)$ edges.
\end{clm}

\begin{proof}
  To satisfy Property~(PR3), $H'_u$ contains an expander graph for the
  pair $(A=L(T_u),B=B'_u)$ described by \lemref{shrink} with parameters
  $k=1/\epsilon$ and $\tau=\log_{1+\epsilon} n/\epsilon$.  For constant
  $\epsilon >0$, this graph has $O(|B'_u|\log n\log\log n)$ edges.

  To satisfy Property~(PR4), $H'_u$ contains an expander graph for the
  pair $(L(T_u),L(T_u))$ described by \lemref{expand} with parameters
  $k=\ell=1/\epsilon$.  For constant $\epsilon >0$, this graph has 
  $O(|T_u|)$ edges.
\end{proof}

Let $G_W$ denote the graph obtained by taking all the edges of $H'_u$
for every special node $u$ in $T$.

\begin{clm}\clmlabel{edges-w}
  The graph $G_W$ has $O(n\log^2 n\log\log n)$ edges.
\end{clm}

\begin{proof}
  The number of all edges used in graphs created to achieve Property~(PR3) is
  \[
      \sum_{u\in S(T)}O(|B_u'|\log n\log\log n)
      = \sum_{i=1}^m O(|B_i|\log n\log\log n)
      = O(n\log^2 n\log\log n)
  \]
  where the final upper bound follows from the convention that $|A_i|\ge
  |B_i|$ and \lemref{wspd-ii}.

  Each graph used to achieve Property~(PR4) for a node $H'_u$ has
  $O(|T_u|)$ edges.  By parititioning the special nodes of $T$ into
  $\ceil{\log_{1+\epsilon} n}$ sets where, for any two nodes $u$ and $u$
  in the same set, $T_u$ and $T_{w}$ are disjoint shows that the total
  number of edges in these graphs is at most $O(n\log n)$.
\end{proof}

%
%

We can now complete the proof of \thmref{main-i}.

\begin{proof}[Proof of \thmref{main-i}]
  Our final construction $G$ contains the graph $G_T$ described in
  \secref{exploding} as well the graph $G_W$ described above.  That $G$
  has $O(n\log^2 n\log\log n)$ edges follows from Claims~\ref{clm:edges-t}
  and \ref{clm:edges-w}.  Given any $F\subseteq P$, we define the set $F^+$
  as follows:
  \begin{enumerate}
     \item Let $D_0=\{u\in V(T): |F^+_T\cap L(T_u)|>(1-3\epsilon)|T_u|\}$.
      [The nodes in $D_0$ are $F^+_T$-dense, so we will abandon
       everything in them.]

     \item For each node $u$ of $T$, let $u'$ denote the root of the largest
      special subtree $T_{u'}$ contained in $T_u$ (possibly $u'=u$) and 
      let $D_1=\{u\in V(T): u' \in D_0\}$. 
      [The nodes in $D_1$ have had their largest special subtree abandoned, so we abandon them as well.] 

     \item Let $F^+_0 = \bigcup_{u\in D_1} L(u)$.
     \item Let $E=S(T)\setminus D_1$ and $F_1^+=\bigcup_{a'\in E}
       S_{H_a'}(F^+_0\cap L(T_u))$. [Each point $q\in F^+_1$ participates
       in a pair $(A_i,B_i)$ with $q\in B_i$, but $q$ has no surviving
       edge to $A_i'$ so we abandon $q$.]
  \end{enumerate}
  Finally, we define $F^+=F^+_0\cup F^+_1$.  Note that $F\subset F^+$
  because $F\subset F^+_T$ and every $x\in F$ is a leaf of $T$ that
  satisfies Condition~1 and therefore Condition~2, so $x\in F^+_0$.

  First we analyze the size of $F^+$.  Let $D'\subseteq D_1$ consist of only those nodes $u$
  such that no ancestor of $u$ is in $D_1$.  Then $L(T_u)$ and $L(T_w)$ are
  disjoint for any distinct nodes $u,w\in D'$ and $F^+_0=\bigcup_{u\in
  D'} L(T_u)$.  
  For each node $u$ in $D'$, 
  \[
      |L(T_u)\cap F^+_T| \ge |L(T_{u'})\cap F^+_T| > (1-3\epsilon)|T_{u'}| 
        \enspace ,
  \]
  so
  \[
      |T_{u'}| < |L(T_u)\cap F^+_T|/ (1-3\epsilon) \enspace .
  \]
  This implies
  \begin{align*}
     |F^+_0| & = \sum_{u\in D'} |F^+_0\cap L(T_u)| \\
             & = \sum_{u\in D'} |T_u| \\
             & \le \sum_{u\in D'} (1+\epsilon)|T_{u'}| \\
             & \le \sum_{u\in D'} (1+\epsilon)|F^+_T\cap L(T_u)|/(1-3\epsilon) \\
             & = (1+\epsilon)|F^+_T|/(1-3\epsilon) \enspace .
  \end{align*}

  Now, the nodes of $E$ can be partitioned into
  $r=\floor{\log_{1+\epsilon} n}+1$ sets $E_1,\ldots,E_r$ where $E_i =
  \{u \in E :\lbl(u)=i\}$.  For any two distinct nodes $u,w\in E_i$,
  $L(T_u)$ and $L(T_w)$ are disjoint, so
  \begin{align*}
     |F^+_1| 
     & = \left|\bigcup_{u\in E} S_{H_u}(F^+_0\cap L(T_u))\right| \\
     & = \left|\bigcup_{i=1}^r\bigcup_{u\in E_i} S_{H_u}(F^+_0\cap L(T_u))\right| \\
     & \le \sum_{i=1}^r\sum_{u\in E_i} |S_{H_u}(F^+_0\cap L(T_u))| \\
     & \le \sum_{i=1}^r\sum_{u\in E_i} \epsilon|F^+_0\cap L(T_u)|/\log_{1+\epsilon} n & \text{(By Property~(PR3))} \\
     & \le \sum_{i=1}^r \epsilon|F^+_0|/\log_{1+\epsilon} n \\
     & \le \epsilon|F^+_0|
  \end{align*}
  Combining the two preceding bounds \clmref{f-plus} we get
  \[
      |F^+|\le |F^+_0|+|F^+_1|
      \le (1+\epsilon)|F^+_0|
      \le (1+\epsilon)^2|F^+_T|/(1-3\epsilon)
      \le (1+\epsilon)^3|F|/(1-3\epsilon)
      \le (1+7\epsilon)|F| \enspace ,
  \]
  for $\epsilon \le \sqrt{145}-12$.

  Next we show that $G-F$ is a $t$-spanner of $P\setminus F^+$.
  Consider any two distinct points $p,q\in P\setminus F^+$ and
  let $(A_i,B_i)\in W$ be the pair such that $p\in A_i$ and $q\in B_i$.

  Since $q$ is not in $F^+_1$, $G_W$ contains an edge $qq'$ with $q'\in
  A_i'\setminus F^+_0$.
  Since $q'$ is not in $F^+_0\supseteq F^+_T$, $q'$ is able to explode
  into $a_i'$.  Specifically, \clmref{explode} implies that there is
  a subset $X_{q'}\subseteq A_i'$, with
  \[  |X_{q'}|\ge (1-a)|A_i'| - |F^+_T\cap A_i'| 
        \ge (1-a)|A_i'| - (1-3\epsilon)|A_i'| 
        \ge (3\epsilon-a)|A_i'| \ge 2\epsilon|A_i'|
  \] 
  (for $a < \epsilon$)
  such that, for every $y\in X_{q'}$, $G-F$ contains a path from $q'$ to $y$
  of length at most $\diam'(T_{a_i'})\le \diam'(a_i)$.

  Similarly, $p$ can explode into $A_i$, so there exists an analogous
  set $X_p\subset A_i$ of size \[  |X_p| \ge (3\epsilon-a)|A_i|
  \ge 2\epsilon|A_i| \enspace .  \] for $a\le\epsilon$.  Now,
  since $|A_i\setminus A_i'|\le \epsilon |A_i|$, 
  \[ |X_p\cap A_i'|\ge |X_p|-|A_i\setminus A_i'|\ge \epsilon|A_i| \ge \epsilon|A_i'| \enspace . \]

  By Property~(PR4) of $H_{a_i'}$, $G_W$ contains an edge $xy$ with $x\in X_p\cap A_i'$ and $y\in X_{q'}$.  
   This yields a path from $p$ to $q$ of length at most
  \begin{align*}
    \dist_{G-F}(p,q) & \le \dist_{G-F}(p,x) + \dist(x,y) + \dist_{G-F}(y,q') + \dist(q',q) \\
      & \le (2C+1)\diam'(a_i) + \dist(q',q) \\
      & \le (2C+2)\diam'(a_i') + \diam'(b_i) + \dist(p,q) \\
      & \le (1+O(1/s))\dist(p,q) \enspace . 
  \end{align*}
  Choosing $s = c/\epsilon$ for a sufficiently large constant $c$
  completes the proof.
\end{proof}

\section{Discussion}

The obvious open problem left by this work is to close the gap between the
$\Omega(n\log n)$ lower bound and our $O(n\log^2\log\log n)$ upper bound
on the number of edges required in a $(1+\epsilon)k$-robust $t$-spanner.

A less obvious object of study is the relation between robustness and
various notions of dimension.  Bose \etal\ show that, any $f(k)$-robust
spanner of the 1-dimensional point set $P_1=\{1,\ldots,n\}$ requires a
superlinear number of edges, for any non-decreasing function $f:\N\to\N$.
For example, they show that, for any constants $\delta>0$ and $t>1$, any
$O(k^{1+\delta})$-robust $t$-spanner of $P$ must have $\Omega(n\log\log
n)$ edges.

The obvious 2-dimensional generalization of $P_1$ is the
$\sqrt{n}\times\sqrt{n}$ grid $P_2=\{1,\ldots,\floor{\sqrt{n}}\}^2$.
However, it turns out that the square grid on $P_2$ is a
$O(k^2)$-robust $3$-spanner with a linear number of edges
(\cite[Section~4]{bose.dujmovic.ea:robust}).  More generally, for any
constant $d$, the $d$-dimensional cubic grid is an $O(k^{d/(d-1)})$-robust
$O(\sqrt{d})$-spanner with a linear number of edges.

The preceding discussion suggests that robustness is a property that is
more easily achieved as the ``true'' dimension of a point set increases,
and this idea is worth exploring.  As a concrete question in this vein,
we propose the following:  Let $P\subset[0,1]^2$ be an $n$-point set
with the property that, for any square $S\subset [0,1]$,
\[
    \floor{an\cdot\ar(S)} \le |P\cap S| \le An\cdot\ar(S)
\]
for some constants $a < 1 < A$.  Determine the slowest-growing function
$f:\N\to\N$ such that $P$ has an $f(k)$-robust spanner with $O(n)$
edges.  Binning the points of $P$ into grid cells of area $1/(an)$ and
use the result described above shows that $f(k)=O(k^2)$.  Is this the
best possible?

\section*{Acknowledgement}

Pat Morin would like to thank Michiel Smid for pointing him to Callahan's thesis, and then pointing him to Chapter~4, and then reading it for him.

\bibliographystyle{plain}
\bibliography{robust2}

\end{document}